\documentclass[letterpaper, 10 pt, conference]{ieeeconf}
\IEEEoverridecommandlockouts
\usepackage{booktabs} 
\usepackage{cite}
\usepackage{algorithmic}
\usepackage[algoruled]{algorithm2e}
\usepackage{textcomp}
\usepackage{xcolor}
\usepackage{bm}
\usepackage{amsfonts}
\usepackage{amsmath,cases,amssymb}
\usepackage[english]{babel}
\usepackage{cite}
\usepackage{color}
\usepackage{epsfig}
\usepackage{nomencl}
\usepackage{url}
\usepackage{bm}
\usepackage{multirow}
\usepackage{hyperref}
\usepackage{color}
\usepackage{multirow}
\usepackage{mathtools}
\usepackage{etoolbox}

\newtheorem{mypro}{Proposition}

\newtheorem{mydef}{Definition}

\newtheorem{myremark}{Remark}

\renewcommand{\hat}{\widehat}

\begin{document}

\title{\LARGE \bf
Learning with Adaptive Conservativeness for Distributionally Robust Optimization: Incentive Design for Voltage Regulation 

}

\author{Zhirui Liang, Qi Li, Joshua Comden, Andrey Bernstein, Yury Dvorkin
\thanks{Zhirui Liang, Qi Li, and Yury Dvorkin are with the Johns Hopkins University, Baltimore MD, U.S. (\{zliang31, qli112, ydvorki1\}@jhu.edu).
Joshua Comden and Andrey Bernstein are with the National Renewable Energy Laboratory, Golden CO, U.S. (\{joshua.comden, andrey.bernstein\}@nrel.gov).}%
\thanks{This work was supported in part by the Risk Aware Power System Control, Dispatch and Market Incentives through ARPA-E PERFORM under Award DE-AR0001300; in part by NSF under Award 2330450, Award 2245507, and Award 2245506. }%
\thanks{This work was authored in part by the National Renewable Energy Laboratory (NREL), operated by Alliance for Sustainable Energy, LLC, for the U.S. Department of Energy (DOE) under Contract No. DE-AC36-08GO28308. This work was supported by the Laboratory Directed Research and Development (LDRD) Program at NREL. The views expressed in the article do not necessarily represent the views of the DOE or the U.S. Government. The U.S. Government retains and the publisher, by accepting the article for publication, acknowledges that the U.S. Government retains a nonexclusive, paid-up, irrevocable, worldwide license to publish or reproduce the published form of this work, or allow others to do so, for U.S. Government purposes.}%
}
\maketitle

\begin{abstract} 
Information asymmetry between the Distribution System Operator (DSO) and Distributed Energy Resource Aggregators (DERAs) obstructs designing effective incentives for voltage regulation. To capture this effect, we employ a Stackelberg game-theoretic framework, where the DSO seeks to overcome the information asymmetry and refine its incentive strategies by learning from DERA behavior over multiple iterations. We introduce a model-based online learning algorithm for the DSO, aimed at inferring the relationship between incentives and DERA responses. Given the uncertain nature of these responses, we also propose a distributionally robust incentive design model to control the probability of voltage regulation failure and then reformulate it into a convex problem. This model allows the DSO to periodically revise distribution assumptions on uncertain parameters in the decision model of the DERA. Finally, we present a gradient-based method that permits the DSO to adaptively modify its conservativeness level, measured by the size of a Wasserstein metric-based ambiguity set, according to historical voltage regulation performance. The effectiveness of our proposed method is demonstrated through numerical experiments.
\end{abstract}

\section{Introduction}
Grid-edge Distributed Energy Resources (DERs), encompassing electricity sources, storage, and demand response programs linked to distribution networks, pose both challenges and opportunities for the Distribution System Operator (DSO) \cite{ding2022unleash}. To tackle issues like forecast uncertainty, monitoring difficulties, and constrained control capabilities, the integration of DER aggregators (DERAs) into a hierarchical control structure has been advocated \cite{zhang2022grid}. This approach allows DSOs to more effectively harness DER flexibility for services like voltage regulation, frequency response, and reserve, though it necessitates well-crafted incentives to boost DER participation. This paper studies the incentive design problem for the DSO, a topic that has gained significant attention \cite{zhou2017incentive, yu2018incentive,yang2023incentivizing}. Specifically, this paper highlights an often-overlooked aspect: DERAs might withhold some information from the DSO, such as the cost and flexibility region of individual DERs, to leverage information asymmetry for profit gains \cite{bialek2021knows}. This necessitates a strategic approach for the DSO to learn the decision-making patterns of DERAs.

Previous research offers various strategies for learning the decision-making models of DERAs, navigating the challenge that the DSO can only observe the responses of DERAs to incentives without direct access to the internal parameters of the DERA decision model. Some studies opt for data-driven and model-free approaches, like the use of neural networks in \cite{zhang2016optimal}, while others apply reverse engineering techniques, such as inverse optimization, to deduce the unseen parameters from observed optimal responses \cite{bian2022demand}. However, these approaches tend to capture a static relationship between incentives and DER responses, overlooking the time-varying nature of parameters within the demand response model.

Several studies have recognized the uncertainty in demand response models and have developed distributionally robust strategies for downstream control tasks. For instance, \cite{mieth2019online, tucker2020constrained} introduced online learning methods to estimate the price sensitivity of demand response participants, incorporating this insight into a distributionally robust optimal power flow model. Similarly, \cite{hassan2020stochastic} formulated distributionally robust control policies using Markov Decision Processes for demand response dispatch. However, these studies typically assume a fixed level of decision conservativeness, which should evolve as parameter learning continues to mirror the increasing confidence in the acquired knowledge. In distributionally robust optimization, this conservativeness level can be quantified by the size of the ambiguity set, as outlined in our previous research \cite{li2024revealing}.


In this paper, we present a model-based online learning algorithm for the DSO to infer the relationship between incentives and DERA responses. Leveraging the learned information, we develop a distributionally robust incentive design model aimed at enhancing voltage regulation in electric power distribution systems. Additionally, we introduce a gradient-based strategy enabling the DSO to adaptively adjust its conservativeness level based on historical voltage regulation performance. This distributionally robust decision-making model with adaptive conservativeness level differs from the \textit{adaptive distributionally robust optimization} proposed in \cite{bertsimas2019adaptive}, which aims to make the best decision for multi-stage problems. It also differs from the decision-aware learning models, e.g.,  uncertainty set learning  as in \cite{wang2023learning,mieth2023prescribed}, which minimize the distance between the expected cost and the best-case cost rather than the worst-case cost as in this paper.  

\section{Overview of Incentive Design Problem}
\label{sec:Preliminary}
\subsection{Incentive-based Voltage Regulation}
The process of incentive-based voltage regulation is characterized by a series of strategic decisions. Initially, the DSO assesses the voltage at all nodes in the distribution system, and then computes incentives for a desired DER response in areas with voltage violations. Then, the DERAs adjust DER settings, leveraging comprehensive control and monitoring over the DER output. Often, this is an iterative process since DER responses will influence voltages and thus the optimal incentives in the future.

This incentive-based voltage regulation problem can be modeled with a Stackelberg game-theoretic framework, which captures the interactions between two players in a leader-follower dynamic. Fig.~\ref{fig:hierarchical_control} shows these interactions for a distribution system with $N_b$ nodes, where the DSO, acting as the single leader, initiates by offering incentives, while the DERAs, as multiple followers in the game, strategically respond by re-dispatching the DER operations. In this paper, we assume a single DERA per node for clarity, though real-world scenarios might feature multiple DERAs per node or a single DERA managing DERs across different nodes. 
\begin{figure}[htbp]
    \centering
    \includegraphics[width=1\linewidth]{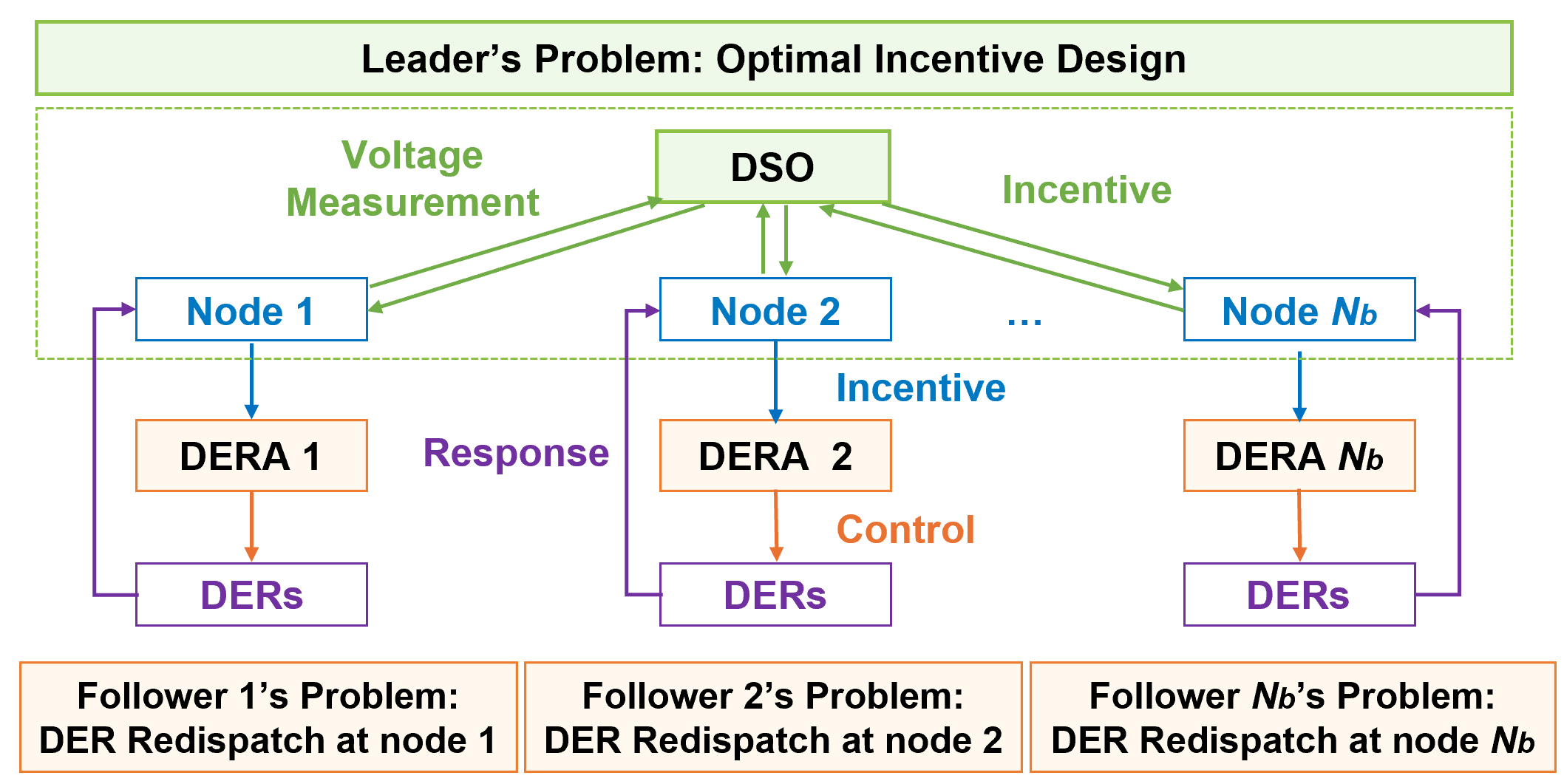}
    \caption{Incentive-based voltage regulation modeled as a Stackelberg game between the DSO (leader) and DERAs (followers). }
    \label{fig:hierarchical_control}
\end{figure}

\subsection{Leader's Problem: Optimal Incentive Design}
The decision-making problem for the DSO is designing incentives to achieve voltage regulation in the most efficient way, which can be formulated as:
\begin{subequations}
\begin{align}
&\min_{\substack{ k_{n}^P,k_{n}^Q, \Delta V_n }}  \   f = \sum\nolimits_{n=1}^{N_b} \left( k_n^{P} \Delta P_{n} + k_n^{Q} \Delta Q_{n} \right) \label{UL:obj}\\
& \mathrm{s.t.} \quad   n= 1,\dots, N_b \nonumber \\ 
    &\quad \Delta V_{n} = \sum\nolimits_{m=1}^{N_b} \alpha_{n,m} \Delta P_{m} + \sum\nolimits_{m=1}^{N_b} \beta_{n,m} \Delta Q_{m} \label{UL:V_P_Q} \\
    &\quad \Delta V_{n}^{\min} \le \Delta V_{n} \le \Delta V_{n}^{\max} \label{UL:V_bounds} \\
    &\quad \Delta P_{n}, \Delta Q_{n} \in \arg\!\min g(\Delta P_{n}, \Delta Q_{n}|k_{n}^P,k_{n}^Q), \label{UL:argmin} 
\end{align}%
\label{mod:UL_DSO}%
\end{subequations}%
\allowdisplaybreaks[0]%
where the decision variables $k_{n}^P$ and $k_{n}^Q$ are the incentives for active and reactive power response at node $n$. 
Objective \eqref{UL:obj} aims to minimize the cost associated with voltage regulation. 
This cost is calculated by summing the products of the incentives and the DERA responses at each node, with $\Delta P_{n}$ and $\Delta Q_{n}$ representing the total active and reactive DER responses at node $n$, as determined by \eqref{UL:argmin}. We assume that the DERs are capable of providing bidirectional responses, meaning the values of $k_{n}^P$, $k_{n}^Q$, $\Delta P_{n}$, and $\Delta Q_{n}$ can be either positive or negative. The auxiliary variable $\Delta V_n$ denotes the voltage change at node $n$, which can be derived from the linearized power flow equation in \eqref{UL:V_P_Q}. In this equation, $\alpha_{n,m}$ and $\beta_{n,m}$ are coefficients linking the voltage change at node $n$ to the DERA response at node $m$ \cite{bernstein2017linear}. Furthermore, \eqref{UL:V_bounds} specifies a permissible range for voltage changes, which is time-varying based on voltage measurements in the system. The response of each DERA is determined by \eqref{UL:argmin}, where  function $g$ will be introduced in the following subsection. 

\subsection{Follower's Problem: DER Redispatch}
Considering the incentives provided by the DSO, along with the operational flexibility and associated costs of each DER, the DERA can identify which set of DERs to re-dispatch, aiming at maximizing its overall profit. The decision-making model for the DERA at node $n$ is:
\begin{subequations}
\begin{align}
\min_{\substack{ \Delta P_{r},\Delta Q_{r}}}  \  & g= \sum\nolimits_{r=1}^{R_n} s_{r} \Delta P_{r} - k_n^{P} \Delta P_{n} - k_n^{Q} \Delta Q_{n} \label{LL:obj}\\
\mathrm{s.t.}  \quad &  r= 1,\dots, R_{n} \nonumber \\ 
  & \Delta P_{r}^{\min} \le \Delta P_{r} \le \Delta P_{r}^{\max} \label{LL:P_limits}\\
  & \phi_r^{\min} \Delta P_{r} \le \Delta Q_{r} \le \phi_r^{\max} \Delta P_{r} \label{LL:P_Q_region} \\
  & \Delta P_{n} = \sum\nolimits_{r=1}^{R_n} \Delta P_{r} \label{LL:P_total} \\
  & \Delta Q_{n} = \sum\nolimits_{r=1}^{R_n} \Delta Q_{r}\label{LL:Q_total} \\
  & k_{n}^P,k_{n}^Q \in \arg\!\min f(k_{n}^P,k_{n}^Q|\Delta P_{n}, \Delta Q_{n}), \label{LL:argmin} 
\end{align}%
\label{mod:LL_DERA}%
\end{subequations}%
\allowdisplaybreaks[0]%
where $R_n$ is the total number of DERs at node $n$, $\Delta P_{r}$ and $\Delta Q_{r}$ are the active and reactive power response of the $r$-th DERs at node $n$, respectively. For simplicity, it is assumed that only $\Delta P_{r}$ incurs a specific unit cost $s_{r}$, whereas $\Delta Q_{r}$ is a by-product of $\Delta P_{r}$ and does not entail any direct cost. The linear feasible regions of $\Delta P_{r}$ and $\Delta Q_{r}$ are enforced by \eqref{LL:P_limits} and \eqref{LL:P_Q_region}, respectively, while the aggregated DER responses at node $n$ are calculated by \eqref{LL:P_total} and \eqref{LL:Q_total}.

\section{Learning for Incentive Design}
\label{sec:learning}
If the DSO has complete knowledge about the decision-making models used by DERAs, including the cost and flexibility of each DER, the optimal incentive design problem could be framed as a bi-level optimization problem, where the upper level model is \eqref{mod:UL_DSO} and the lower level model is \eqref{mod:LL_DERA}. However, the limited knowledge available to the DSO complicates the direct implementation of this bi-level optimization strategy  and necessitates a learning approach. In this section, we introduce an online learning algorithm, where the DSO can infer the underlying decision-making models of DERAs by continuously observing DERA responses to incentives within the Stackelberg game context.


\subsection{Learning with Surrogate Model}
The DSO seeks to identify the mapping from the incentives $k_n^P, k_n^Q$ to the responses $\Delta P_{n},\Delta Q_{n}$ through a learning process and to use this mapping in \eqref{UL:argmin}. The relationship between incentives and responses is essentially determined by the decision-making model of the DERA detailed in \eqref{mod:LL_DERA}. However, the DSO cannot directly infer the model in \eqref{mod:LL_DERA} due to two main factors. First, the DSO has access only to the aggregated demand response at each node ($\Delta P_n$ and $\Delta Q_n$) without visibility into individual DER responses ($\Delta P_r$ and $\Delta Q_r$). Second, several parameters in \eqref{mod:LL_DERA}, including the costs $s_r$ and the feasible ranges for $\Delta P_r$ and $\Delta Q_r$ determined by \eqref{LL:P_limits} and \eqref{LL:P_Q_region}, are time-varying due to the variability of DER resources. Therefore, the DSO has to learn this mapping relationship based on a surrogate model. Without losing generality (as discussed in Section \ref{sec:prop} below), we assume the DSO employs the following model to approximate \eqref{mod:LL_DERA}:
\begin{subequations}
\begin{align}
   \min_{\substack{ \Delta P_{n},\Delta Q_{n}}}  \quad  & c_{n} \Delta P_{n}^2 + d_{n} \Delta P_{n}- k_n^{P} \Delta P_{n} - k_n^{Q} \Delta Q_{n} \label{surrogate:obj}\\
\text{s.t. } \quad 
   &  \Delta Q_{n} = \phi_{n} \Delta P_{n} \label{surrogate:constraint}.
\end{align}%
\label{mod:surrogate}%
\end{subequations}%
\allowdisplaybreaks[0]%
where $c_{n}$ and $d_n$ are the coefficients of the quadratic and linear terms, respectively, in the aggregated cost function at node $n$, while $\phi_{n}$ is the ratio of the total reactive to active power responses at node $n$. Given the time-varying nature of values in the set of $\mathcal{S}=\{c_{n,t},d_{n,t},\phi_{n,t}\}_{n=1}^{N_b}$, they can be modeled as uncertain parameters, with their distribution characterized by samples. We will then demonstrate the process for deriving a set of samples for these uncertain parameters using observations at time $t$, denoted by $\{\Delta P_{n,t},\Delta Q_{n,t},k_{n,t}^P, k_{n,t}^Q\}_{n=1}^{N_b}$.

According to \eqref{surrogate:constraint}, $\phi_{n,t}$ can be calculated directly by:
\begin{align}
    \phi_{n,t} = {\Delta Q_{n,t}}/{\Delta P_{n,t}} \label{mod:estimate_phi}
\end{align}
By replacing $\Delta Q_{n,t}$ with $\phi_{n,t} \Delta P_{n,t}$ in \eqref{mod:surrogate}, we transform it into an unconstrained optimization problem. For optimality, it requires the following first-order condition to hold:
\begin{align}
   2 { c_{n,t}} \Delta P_{n,t} +  d_{n,t} - k_{n,t}^{P} - k_{n,t}^{Q} \phi_{n,t} =0 \label{mod:estimate_c}
\end{align}
However, \eqref{mod:estimate_c} alone does not suffice to uniquely determine the specific values of $c_{n,t}$ and $d_{n,t}$. We need another time instance $t'$ to derive an additional equation involving $c_{n,t'}$ and $d_{n,t'}$. By assuming that $c_{n,t} = c_{n,t'}$ and $d_{n,t} = d_{n,t'}$, we can calculate their values using these two equations. With this new set of estimated parameter values, the DSO can refine the distributions of $c_{n}$, $d_{n}$, and $\phi_{n}$ by incorporating these new samples into their historical sample sets.

Consider a sample set of $\{\mathcal{S}\}_{N_s}$ with $N_s$ samples for each uncertain parameter. Based on each sample, we can derive the optimal response of the DERA at node $n$ as a linear function of the incentives at node $n$ based on the Karush-Kuhn-Tucker (KKT) conditions of \eqref{mod:surrogate}:
\begin{align}
   \Delta P_{n,i} &= \frac{1}{2 c_{n,i}}k_n^P + \frac{\phi_{n,i}}{2 c_{n,i}} k_n^Q - \frac{d_{n}}{2 c_{n,i}},\ i=1,...,N_s  \label{relation_P} \\
   \Delta Q_{n,i} & = \frac{\phi_{n,i}}{2 c_{n,i}}k_n^P + \frac{\phi_{n,i}^2}{2 c_{n,i}} k_n^Q - \frac{\phi_{n,i} d_{n,i}}{2 c_{n,i}},\ i{\rm{=}}1,...,N_s \label{relation_Q} 
\end{align}

\subsection{Properties of Surrogate Model} \label{sec:prop}
The selection of the surrogate model in \eqref{mod:surrogate} is justified by three key reasons:
\begin{enumerate}
    \item The surrogate model \eqref{mod:surrogate} preserves certain characteristics of the original DERA decision model in \eqref{mod:LL_DERA}. For instance, the quadratic function in \eqref{surrogate:obj} is a good approximation of the piece-wise linear function with an incrementally increasing slope in \eqref{LL:obj}.
    \item The surrogate model \eqref{mod:surrogate} features unknown parameters that can be accurately inferred using a relatively limited set of observations, i.e., through \eqref{mod:estimate_phi} and \eqref{mod:estimate_c}.
    \item The surrogate model \eqref{mod:surrogate} facilitates the derivation of an explicit mapping relationship from $k_n^P, k_n^Q$ to $\Delta P_{n},\Delta Q_{n}$, as shown in \eqref{relation_P} and \eqref{relation_Q}.
\end{enumerate}
We recognize that \eqref{mod:surrogate} is not the only option for a surrogate model. Any model meeting the aforementioned three criteria could serve as an effective surrogate of \eqref{mod:LL_DERA} and be used in the learning process. 

The surrogate model \eqref{mod:surrogate}, while sharing characteristics with the actual DERA decision model in \eqref{mod:LL_DERA}, diverges in terms of modeling granularity. Specifically, \eqref{mod:surrogate} operates at the node level, whereas \eqref{mod:LL_DERA} focuses on individual DERs. This difference raises questions about the effectiveness of the proposed learning algorithm in capturing the relationship between incentives and DERA responses.

We claim that, despite the difference in granularity, \eqref{mod:surrogate} can effectively capture the essential relationship between $k_n^P, k_n^Q$ and $\Delta P_{n},\Delta Q_{n}$ determined by \eqref{mod:LL_DERA}. This effectiveness is primarily because the DSO is concerned with the aggregated response at each node, rather than individual DER responses. For instance, if the constraint of $\Delta P_r \le \Delta P_r^{\max}$ is not binding, variations in $\Delta P_r^{\max}$ will not impact the aggregate responses $\Delta P_{n}$ and $\Delta Q_{n}$. Define the set of uncertain parameters in \eqref{mod:LL_DERA} as $\mathcal{S}^\text{full} = \{c_{r},\Delta P_{r}^{\min},\Delta P_{r}^{\max}, \phi_{r}^{\min}, \phi_{r}^{\max}, r=1,...R_n, n = 1,...,N_b \}$. Essentially, the distribution of $\mathcal{S}$ originates from a specific subset of the distribution of $\mathcal{S}^\text{full}$, but only this subset is vital for the mapping from $k_n^P, k_n^Q$ to $\Delta P_{n},\Delta Q_{n}$. Consequently, the distribution of $\mathcal{S}$ captures all the variations in $\mathcal{S}^\text{full}$ that matter to the DSO.

Additionally, the effectiveness of the proposed learning strategy can be further improved. For instance, instead of directly learning distributions $P(c_{n})$, $P(d_{n})$, and $P(\phi_{n})$, it may be more beneficial to learn conditional distributions $P(c_{n}|\Psi_n)$, $P(d_{n}|\Psi_n)$, and $P(\phi_{n}|\Psi_n)$, where $\Psi_n$ includes all pertinent external factors like time and temperature that could affect the availability of flexible DERs. 

\section{Distributionally Robust Incentive Design}
\label{sec:DRO}
By replacing \eqref{UL:argmin} with \eqref{relation_P} and \eqref{relation_Q}, we convert \eqref{mod:UL_DSO} into a single-level robust optimal incentive design model that is more computationally tractable. However, employing robust optimization strategies can lead to overly conservative decisions. Additionally, future instances of uncertain parameters may deviate from their historical patterns, potentially leading to unforeseen adverse scenarios. In response, this section introduces a distributionally robust incentive design approach that accounts for the empirical distribution of uncertainties based on all available samples and accommodates deviations within a predetermined range from the empirical distribution. We model voltage limits as a chance constraint, which effectively manages the risk of voltage level violations at each node. We also provide a convex reformulation of this distributionally robust incentive design model, facilitating its solution with standard optimization solvers.

\subsection{Distributionally Robust Incentive Design Model}
To streamline our notation, we introduce a new vector $\xi_n$ for node $n$, encapsulating five uncertainties derived from the three initial uncertain parameters $c_{n},d_{n},\phi_{n}$ in \eqref{mod:surrogate}:
\begin{equation}
    \xi_{n} {\rm{=}} \left[\xi_{n,1},...,\xi_{n,5}\right] {\rm{=}}\left[\frac{1}{2 c_{n}},\frac{\phi_{n}}{2 c_{n}}, \frac{{\rm{-}}d_{n}}{2 c_{n}}, \frac{\phi_{n}^2}{2 c_{n}},  \frac{{\rm{-}}\phi_{n} d_{n}}{2 c_{n}}\right] \nonumber
\end{equation}
The distributionally robust incentive design model is: 
\begin{subequations}
\begin{align}
& \inf_{\substack{ k_{n}^P,k_{n}^Q}}  \ \sup_{\mathbb{P} \in \mathcal{A}} \ \mathbb{E}^{\mathbb{P}} \left[ \sum\nolimits_{n=1}^{N_b} \left(k_n^{P} \Delta P_{n} + k_n^{Q} \Delta Q_{n} \right) \right] \label{uncertain:obj}\\
&\mathrm{s.t.} \quad  n= 1,\dots, N_b \nonumber \\
    & \quad \Delta P_{n}= \xi_{n,1} k_n^P + \xi_{n,2} k_n^Q + \xi_{n,4} \label{uncertain_P_calculation} \\
    & \quad \Delta Q_{n}= \xi_{n,2} k_n^P + \xi_{n,3} k_n^Q + \xi_{n,5} \label{uncertain_Q_calculation}\\
    & \quad \Delta V_{n} = \sum\nolimits_{m =1}^{N_b} \alpha_{n,m} \Delta P_{m} + \sum\nolimits_{m =1}^{N_b} \beta_{n,m} \Delta Q_{m}\label{uncertain:V} \\
    & \quad \inf_{\mathbb{P} \in \mathcal{A}} \ \mathbb{P}\left\{
         \begin{array}{c}
              \Delta V_n^{\min}  \le \Delta V_{n} \\
              \Delta V_{n} \le \Delta V_{n}^{\max} 
         \end{array},n{\rm{=}} 1,..., N_b \right\} \geq 1{\rm{-}}\gamma \label{uncertain:V_bound} 
\end{align}%
\label{mod:uncertain}%
\end{subequations}%
\allowdisplaybreaks[0]%
This model assumes that the DSO adopts a risk-averse stance regarding the voltage regulation problem. The objective in \eqref{uncertain:obj} aims to minimize the voltage regulation cost in the worst-case scenario. The joint chance constraint in \eqref{uncertain:V_bound} guarantees that the uncertain value of $\Delta V_{n}$ remains within the prescribed limits, $\Delta V_n^{\min}$ and $\Delta V_n^{\max}$, with a minimum probability of $1-\gamma$. The distribution of $\xi_n$ is permitted to vary within an ambiguity set $\mathcal{A}$. In this paper, we construct this ambiguity set through the application of the Wasserstein metric due to its superior out-of-sample performance \cite{mohajerin2018data}.

\subsection{Wasserstein Metric-based Ambiguity Set} 
We first give the definition of the Wasserstein metric, a measure of distance between two probability distributions:
\begin{mydef}    
(\textbf{Wasserstein Distance}) The Wasserstein metric $d_{W}(\mathbb{P}_1, \mathbb{P}_2): \mathcal{M}(\Xi) \times \mathcal{M}(\Xi) \rightarrow \mathbb{R}$ is defined by:
    \begin{equation} \nonumber
        d_{W}(\mathbb{P}_1,\mathbb{P}_2) = \inf_{\pi\in \Pi(\mathbb{P}_1,\mathbb{P}_2)}\left\{ \int_{\Xi^2}\Vert \xi_1-\xi_2 \Vert\pi(d\xi_1,d\xi_2) \right\},
    \end{equation}
where $\Xi$ is the domain of a random variable $\xi$ (i.e., $\xi \in \Xi$), $\mathcal{M}(\Xi)$ contains all distributions defined on $\Xi$, $\Pi(\mathbb{P}_1,\mathbb{P}_2)$ is the set of joint probability distributions of $\xi_1,\xi_2$ with marginal distributions of $\mathbb{P}_1$ and $\mathbb{P}_2$ \cite{mohajerin2018data}.
\end{mydef}

Given $N_s$ historical samples $\{ \hat{\xi}_i\}_{i=1}^{N_s}$, the empirical distribution $\hat{\mathbb{P}}_{N_s}$ can be derived as:
\begin{equation} \label{empirical distribution}
    \hat{\mathbb{P}}_{N_s} = \frac{1}{N_s} \sum\nolimits_{i=1}^{N_s} \delta_{\hat{\xi}_i},
\end{equation}
where $\delta_{\hat{\xi}_i}$ is a Dirac distribution centered at $\hat{\xi}_i$. The ambiguity set based on these historical samples is defined as a Wasserstein ball centered at $\hat{\mathbb{P}}_{N_s}$: 
\begin{equation} \label{ambiguity set}
    \mathcal{B}_{\epsilon}(\hat{\mathbb{P}}_{N_s}) = \left\{ \mathbb{P} \in \mathcal{M}(\Xi) \left\vert d_{W}\left(\mathbb{P}, \hat{\mathbb{P}}_{N_s}\right) \right.\leq \epsilon \right\},
\end{equation}
where $\epsilon$ is the radius of the Wasserstein ball. Here we treat $\epsilon$ as a predefined value. In Section~\ref{sec:Adaptive_DRO}, we will study the method for adaptively adjusting the value of $\epsilon$ to optimize performance in voltage regulation.

\subsection{Convex Formulation of Distributionally Robust Model}
The distributionally robust incentive design model in \eqref{mod:uncertain} presents challenges for direct solution until it is converted into a convex optimization problem. This conversion leverages the conditional value-at-risk (CVaR) and the Wasserstein metric-based ambiguity set. To maintain clarity and conciseness, we omit the detailed steps of this reformulation. Readers seeking a deeper understanding of these steps are directed to \cite{mohajerin2018data} for further information.

For leaner notation in the following derivations, we collect all decision variables in vector $x \in \mathbb{R}^{6 N_b}$ and all uncertainties in vector $\xi \in \mathbb{R}^{6 N_b}$, i.e., 
\begin{align}
    &x = \large[(k_1^P)^2,\dots, (k_{N_b}^P)^2, 2k_1^P k_1^Q, \dots, 2k_{N_b}^P k_{N_b}^Q, \nonumber \\
    & \qquad (k_1^Q)^2 ,\dots, (k_{N_b}^Q)^2, k_1^P,\dots, k_{N_b}^P, k_1^Q ,\dots, k_{N_b}^Q ,\bm{1}_{N_b}\large]^{\top}\nonumber \\
    &\xi = \left[\xi_{1,1},\dots, \xi_{{N_b},1}, \dots, \xi_{1,5},\dots, \xi_{{N_b},5},\bm{0}_{N_b}\right]^{\top}, \nonumber
\end{align}
where $\bm{1}_{N_b}$ and $\bm{0}_{N_b}$ denote all-one and all-zero vectors of length ${N_b}$. Then, we can re-write \eqref{mod:uncertain} as:
\begin{subequations} \label{mod:general_DRO}
    \begin{align}
        \inf_{x \in \mathcal{F}}  \ & \sup_{\mathbb{P} \in \mathcal{B}_{\epsilon}} \mathbb{E}^{\mathbb{P}} \left[\langle x, \xi \rangle \right] \label{general_DRO:obj}\\
        \mathrm{s.t.} \
         & \inf_{\mathbb{P} \in \mathcal{B}_{\epsilon}}  \mathbb{P} \left\{ \max_{k=1,...,2{N_b}} \left[  a_{k}^{\prime} \xi {\rm{+}} b_{k}^{\prime} \right] \leq 0\right\} \geq 1 {\rm{-}} \gamma, \label{general_DRO:constraint}
    \end{align}
\end{subequations}
where $\mathcal{F}$ is the feasible space of $x$, 
and $\mathbb{P}$ is the possible distribution of $\xi$ constrained by a Wasserstein ball $\mathcal{B}_{\epsilon}$.
In \eqref{general_DRO:constraint}, $a^{\prime}_k$ is the $k$-th row of this $2N_b\times 6N_b$ matrix:
\begin{align}
    \nonumber
    &\begin{bmatrix}
    ({\rm{-}}W_1 x)^{\top} \\
    \vdots \\
    ({\rm{-}}W_{N_b} x)^{\top} \\
    (W_1 x)^{\top} \\
    \vdots \\
    (W_{N_b} x)^{\top}
    \end{bmatrix}\text{, } 
    W_n =
        \begin{bmatrix} 
        {{\mathcal{O}_{6 \times 3}}}&{\begin{matrix}
        A_n & \mathcal{O}_{1\times 1} & \mathcal{O}_{1\times 1}\\
        B_n & A_n & \mathcal{O}_{1\times 1} \\
        \mathcal{O}_{1\times 1} & B_n & \mathcal{O}_{1\times 1} \\
        \mathcal{O}_{1\times 1} & \mathcal{O}_{1\times 1} & A_n \\
        \mathcal{O}_{1\times 1} & \mathcal{O}_{1\times 1} & B_n \\
        \mathcal{O}_{1\times 1} & \mathcal{O}_{1\times 1} & \mathcal{O}_{1\times 1}
        \end{matrix}}
        \end{bmatrix}  \nonumber\\
    &A_n = \mathcal{D}(\alpha_{n,1},\dots,\alpha_{n,{N_b}}),\
    B_n = \mathcal{D}(\beta_{n,1},\dots,\beta_{n,{N_b}}), \nonumber 
\end{align}
where $\mathcal{D}(\cdot)$ creates a diagonal matrix from a vector, $\mathcal{O}_{m\times n} $ is a all-zero matrix with a dimension of $mN_b \times n N_b$. Similarly, $b^{\prime}_k$ is the $k$-th entry of the following $2N_b\times 1$ vector:
\begin{equation}
    \left[\Delta V_1^{\min},\dots, \Delta V_{N_b}^{\min}, -\Delta V_1^{\max},\dots, -\Delta V_{N_b}^{\max}\right]^{\top} \nonumber
\end{equation}
Following \cite{mohajerin2018data}, we can transform \eqref{mod:general_DRO} into the following convex problem:
\begin{subequations} \label{mod:convex_DRO}
    \begin{align}
        & \min_{x,\tau,\lambda^{co}_{j},\lambda^{cc}_{j},s^{co}_{ji},s^{cc}_{ji}} \quad \sum\nolimits_{j=1}^{6N_b} \left( \lambda^{co}_{j} \epsilon {\rm{+}} \frac{1}{N_s} \sum\nolimits_{i=1}^{N_s} s^{co}_{ji}\right) \label{convex_DRO:obj} \\ 
        & \mathrm{s.t.} \ i = 1,\dots,N_{s}, \ j=1,\dots,6N_b, \ k = 0,\dots,2N_b \nonumber \\
        & \quad (\rho_{ij}^{1}) \ s^{co}_{ji}  \ge x_{j} \hat{\xi}_{ji} \label{convex_DRO:obj_1}\\
         &\quad (\rho_{ij}^{2}) \ s^{co}_{ji} \ge x_{j} \overline{\xi}_{j} - \lambda^{co}_{j} \vert \overline{\xi}_{j} - \hat{\xi}_{ji}\vert  \label{convex_DRO:obj_2} \\
         &\quad (\rho_{ij}^{3}) \ s^{co}_{ji} \ge x_{j} \underline{\xi}_{j} - \lambda^{co}_{j} \vert \underline{\xi}_{j} - \hat{\xi}_{ji}\vert \label{convex_DRO:obj_3} \\
         & \quad  (\mu)\ \gamma \tau  +  \sum\nolimits_{j=1}^{6N_b} \left( \lambda^{cc}_{j} \epsilon + \frac{1}{N_s} \sum\nolimits_{i=1}^{N_s} s^{cc}_{ji}\right) \leq 0 \label{convex_DRO:CVaR}\\
         & \quad (\phi_{ikj}^{1}) \ s^{cc}_{ji} \ge  a_{kj} \hat{\xi}_{ji} + b_k \label{convex_DRO:CVaR_1}\\
         &\quad (\phi_{ikj}^{2}) \ s^{cc}_{ji} \ge a_{kj} \overline{\xi}_{j} + b_k - \lambda^{cc}_{j} \vert \overline{\xi}_{j} - \hat{\xi}_{ji}\vert \label{convex_DRO:CVaR_2}\\
         &\quad (\phi_{ikj}^{3}) \ s^{cc}_{ji} \ge  a_{kj} \underline{\xi}_{j} + b_k - \lambda^{cc}_{j} \vert \underline{\xi}_{j} - \hat{\xi}_{ji}\vert \label{convex_DRO:CVaR_3}\\
         & \quad  \lambda^{co}_{j},\ \lambda^{cc}_{j} \geq 0,
    \end{align}%
\end{subequations}%
\allowdisplaybreaks[0]%
Eq.~\eqref{convex_DRO:obj}-\eqref{convex_DRO:obj_3} are the \textit{exact} reformulation of the objective in \eqref{general_DRO:obj}, while  \eqref{convex_DRO:CVaR}-\eqref{convex_DRO:CVaR_3} are the \textit{inner} approximation of the chance constraint in \eqref{general_DRO:constraint}.
This reformulation is derived by introducing $\overline{a}_{0} = \bm{0}_{6N_b}$, $\overline{b}_{0} = 0$, $\overline{a}_{k} = a_{k}^{\prime}$, and $\overline{b}_{k} = b_{k}^{\prime} - \tau$, $k=1,...,2N_{b}$. Then, $a_{kj}$ is the $j$-th coordinate of $\overline{a}_{k}$, $b_k = \overline{b}_k/6N_b$, $\underline{\xi}_j$ and $\overline{\xi}_j$ are the lower and upper bounds, respectively, of the $j$-th coordinate of the domain of $\xi$. Additionally, $\tau$, $\lambda_{j}^{co}$, $\lambda_{j}^{cc}$, $s_{ji}^{co}$ and $s_{ji}^{cc}$ are auxiliary variables.
The superscript in $\lambda_j^{co}$ and $s_{ji}^{co}$ highlights its connection to the cost function reformulation. Similarly, the superscript in $\lambda_j^{cc}$ and $s_{ji}^{cc}$ highlight its connection to the chance-constraint reformulation. Greek letters listed on the left side of \eqref{convex_DRO:CVaR}-\eqref{convex_DRO:obj_3} are dual variables of the corresponding constraints.

Model \eqref{mod:convex_DRO} constitutes a finite-dimensional convex problem that can be efficiently solved using commercial optimization solvers. The number of constraints in \eqref{mod:convex_DRO} correlates directly with $N_s$, which is the number of samples factored into the empirical distribution $\hat{\mathbb{P}}_{N_s}$.

\section{Conservativeness Update in Incentive Design}
\label{sec:Adaptive_DRO}
In the distributionally robust incentive design model \eqref{mod:convex_DRO}, the size of the Wasserstein metric-based ambiguity set, indicated by the value of $\epsilon$, determines the conservativeness of the decision-makers, mirroring their confidence in the empirical distribution of $\xi$. A larger $\epsilon$ incorporates a broader range of possible distributions in the decision making model, leading to decisions that are more conservative. Intuitively, as learning in Section~\ref{sec:learning} progresses, the DSO will obtain an empirical distribution of $\xi$ that increasingly better approximates its true distribution. Consequently, the DSO should become less conservative in incentive design. 

In this section, we propose a method for assessing the suitability of the conservativeness level \textit{post hoc} and a conservativeness updating algorithm to optimize voltage regulation performance. 
Fig.\ref{fig:adaptive} illustrates how the conservativeness level can be updated within the proposed framework. The inner loop depicted in Fig.\ref{fig:adaptive} combines the distributionally robust incentive design model introduced in Section~\ref{sec:DRO} (i.e., the green block in Fig.~\ref{fig:adaptive}) with the DERA decision-making model outlined in Section~\ref{sec:Preliminary} (i.e., the orange block in Fig.~\ref{fig:adaptive}). After executing the inner loop several times, such as over a day or multiple days, the DSO can use the learning approach detailed in Section~\ref{sec:learning} to refine its understanding of the empirical distribution (i.e., the red block in Fig.~\ref{fig:adaptive}). Meanwhile, the DSO can adjust its level of conservativeness using the algorithm described in this section (i.e., the yellow block in Fig.~\ref{fig:adaptive}). This belief update process, or outer loop, is repeated until the empirical distribution of uncertainty  approximates the true distribution acceptably well, and the conservativeness level $\epsilon$ stabilizes at an optimal value.
\begin{figure}[htbp]
    \centering
    \includegraphics[width=1\linewidth]{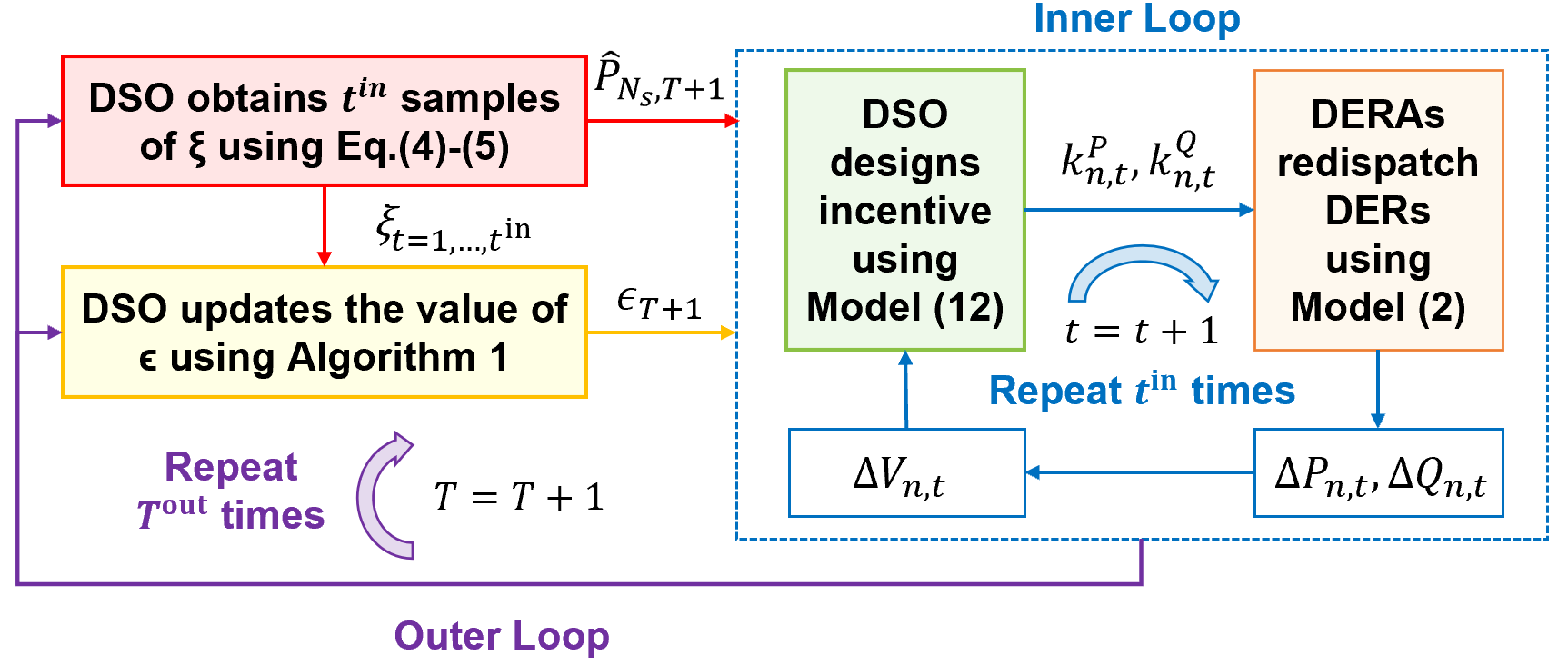}
    \caption{Distributionally robust incentive design for voltage regulation with learning and adaptive conservativeness level. }
    \label{fig:adaptive}
\end{figure}

\subsection{Adequacy of Current Conservativeness Level}
To determine the adequacy of the current level of conservativeness, we need to understand the rationale behind employing a distributionally robust approach for incentive design. The DSO has two objectives when solving the distributionally robust incentive design model in \eqref{mod:uncertain}:
\begin{enumerate}
    \item Objective 1: Minimize the discrepancy between the expected and actual worst-case costs associated with voltage regulation, as outlined in \eqref{uncertain:obj}.
    \item Objective 2: Ensure the expected worst-case voltage violation rate, denoted by $\gamma$ in \eqref{uncertain:V_bound}, aligns closely with the actual occurrence rate of such violations.
\end{enumerate}
In this paper, the conditional value-at-risk (CVaR) is used to measure the violation rate of a chance-constraint. As demonstrated in \cite{rockafellar2000optimization}, enforcing a CVaR $\le 0$ ensures that the corresponding chance constraint is satisfied.

Throughout the time steps in the inner loop, the expected worst-case cost at each specific time $t$, labeled as $\text{Cost}_t^{exp}$, is equal to the objective value of \eqref{convex_DRO:obj} at time $t$. The expected CVaR, denoted as $\text{CVaR}^{exp}$, is a statistical value derived from results over all time steps. Specifically, we have $\text{CVaR}^{exp} = 0$ if the chance constraint in \eqref{uncertain:V_bound} is consistently met in the incentive design process.


One the other hand, the actual cost and CVaR, denoted by $\text{Cost}_t^{act}$ and $\text{CVaR}^{act}$ respectively, are determined by the actual DERA responses over the past $t^\text{in}$ steps:
\begin{align}
    & \text{Cost}_t^{act} {\rm{=}} \sum\limits_{n=1}^{N_b} \left(k_{n,t}^{P} \Delta P_{n,t} {\rm{+}} k_{n,t}^{Q} \Delta Q_{n,t} \right),\ t=1,...,t^\text{in} \label{actual_Cost} \\ 
    & \text{CVaR}^{act} {\rm{=}}  \inf_{\tau} \gamma \tau {\rm{+}} \frac{1}{t^\text{in}} \sum\limits_{t=1}^{t^\text{in}} \left[ \max_{k=0,1,...,2N_{b}} \langle \overline{a}_k, \xi_t \rangle {\rm{+}} \overline{b}_k \right]. \label{actual_CVaR}
\end{align}
Calculating $\text{CVaR}^{act}$ requires solving the optimization problem in \eqref{actual_CVaR} using $\xi_t$, which are the realizations of uncertainty $\xi$ over the past $t^\text{in}$ steps. These new realizations of $\xi$ are estimated by the DSO using the learning method in Section~\ref{sec:learning}.

We establish the following loss function to evaluate the performance of $\epsilon$ over the preceding $t^\text{in}$ time steps: 
\begin{equation}
    \mathcal{L}(\epsilon)= (\text{Cost}_{\underline t}^{exp} - \text{Cost}_{\underline t}^{act})^2+(\text{CVaR}^{exp} - \text{CVaR}^{act})^2, \label{loss_function}
\end{equation}
where $\underline{t}$ denotes the time when the actual worst-case scenario occurs, i.e., when the discrepancy between the expected worst-case cost and the actual cost is minimized:
\begin{equation}
\underline t = {\arg \min}_{t=1,...,t^\text{in}} \vert \text{Cost}_t^{exp} - \text{Cost}_t^{act} \vert \label{worst_case}
\end{equation}
It is sufficient to focus on $ \underline t$ in the loss function \eqref{loss_function} because the actual costs in scenarios other than the worst-case do not reflect the performance of risk-averse incentives and, therefore, should not be employed to update $\epsilon$.

\subsection{Conservativeness Update Algorithm}
Based on the loss function in \eqref{loss_function}, we can use the gradient descent method to find the optimal value of $\epsilon$ following the update rule of:
\begin{equation}
    \epsilon' = \epsilon -\chi \frac{\partial \mathcal{L}(\epsilon)}{\partial \epsilon}  ,\label{updating_rule}
\end{equation}
where $\chi$ is a predefined learning rate. The gradient term in \eqref{updating_rule} can be further expressed as:
\begin{align}
    \frac{\partial \mathcal{L}(\epsilon)}{\partial \epsilon} &=  2(\text{Cost}_{\underline t}^{exp} - \text{Cost}_{\underline t}^{act}) \frac{\partial \text{Cost}(\epsilon)}{\partial \epsilon} \nonumber \\
    & + 2(\text{CVaR}^{exp} - \text{CVaR}^{act}) \frac{\partial \text{CVaR}(\epsilon)}{\partial \epsilon}, \label{gradient}
\end{align}
Following the convex formulation in \eqref{mod:convex_DRO}, the worst-case cost and CVaR can be written as functions of $\epsilon$:
\begin{align}
    & \text{Cost}(\epsilon) = \sum\nolimits_{j=1}^{6N_b} \left( \lambda^{co}_{j} \epsilon + \frac{1}{N_s} \sum\nolimits_{i=1}^{N_s} s^{co}_{ji}\right) \label{function_Cost} \\
    & \text{CVaR}(\epsilon) {\rm{=}} \mu \left[\tau \gamma {\rm{+}}  \sum\nolimits_{j=1}^{6N_b} \left( \lambda^{cc}_{j} \epsilon {\rm{+}} \frac{1}{N_s} \sum\nolimits_{i=1}^{N_s} s^{cc}_{ji}\right) \right] \label{function_CVaR},
\end{align}
where \eqref{function_Cost} is the same as the objective function in \eqref{convex_DRO:obj}, and \eqref{function_CVaR} is the complementary slackness condition based on \eqref{convex_DRO:CVaR}. Therefore, the two partial derivatives in \eqref{gradient} can be calculated as:
\begin{align}
    & \frac{\partial \text{Cost}(\epsilon)}{\partial \epsilon} = \sum\nolimits_{j=1}^{6 N_b}\lambda_{j,\underline t}^{co}, \label{gradient_Cost} \\
    & \frac{\partial \text{CVaR}(\epsilon)}{\partial \epsilon} = \frac{1}{t^\text{in}} \sum\nolimits_{t=1}^{t^\text{in}} \mu_t \sum\nolimits_{j=1}^{6 N_b}\lambda_{j,t}^{cc}, \label{gradient_CVaR}
\end{align}
where $\lambda_{j,t}^{co}$ and $\lambda_{j,t}^{cc}$ are auxiliary variables in \eqref{mod:convex_DRO} and can be obtained directly from solving \eqref{mod:convex_DRO}, while $\mu_t$ is a dual variable in \eqref{mod:convex_DRO} and can be returned directly by some optimization solvers or can be derived by the KKT conditions of \eqref{mod:convex_DRO}. The two partial derivatives in \eqref{gradient_Cost} and \eqref{gradient_CVaR} are always non-negative since \eqref{mod:convex_DRO} requires $\lambda_{j,t}^{co}, \lambda_{j,t}^{co}, \mu_t \ge 0$ for any $j$ and $t$.
Algorithm~\ref{alg:adaptive} summarizes this gradient-based update of $\epsilon$.
\begin{algorithm}[t]
    \small
    \SetAlgoLined
    \SetKwInOut{Input}{input}\SetKwInOut{Output}{output}
    \Input{$\{k_{n,t}^P,k_{n,t}^Q,\Delta P_{n,t},\Delta Q_{n,t}\}_ {n=1,...,N_b,t=1,...t^\text{in}}$;\\
    $\{\lambda_{j,t}^{co},\lambda_{j,t}^{cc},\mu_t\}_{j=1,...,6N_b,t=1,...t^\text{in}}$;\\
    Initial values  $\{\text{cost}_t^{exp}\}_{t=1,...t^\text{in}}$ ;\\
    Current conservativeness level $\epsilon_{T}$;\\
    Maximum iterations allowed $count^{\max}$; \\
    Convergence threshold $\Delta \epsilon^{\min}$\\
    }
    \Output {Updated conservativeness level $\epsilon_{T+1}$;
    }
    \Begin{
      Calculated $\text{cost}^{act}$ and $\text{CVaR}^{act}$ using \eqref{actual_Cost} and \eqref{actual_CVaR}; \\
      Calculated two partial derivatives using \eqref{gradient_Cost} and \eqref{gradient_CVaR}; \\
      $\Delta \epsilon  \leftarrow \epsilon_T$; $\epsilon \leftarrow \epsilon_T$; $count \leftarrow 0$\\ 
      \While{$count < count^{\max}$ and $\Delta \epsilon > \Delta \epsilon^{\min}$ }
      {
      Calculate gradient following \eqref{gradient}; \\
      Obtain $\epsilon'$ following \eqref{updating_rule}; \\
      Re-run \eqref{mod:convex_DRO} for $t^\text{in}$ times with $\epsilon'$ to obtain updated values of $\{\text{cost}_t^{exp}\}_{t=1,...t^\text{in}}$;\\
      $\Delta \epsilon \leftarrow |\epsilon-\epsilon'|$ ;
      $\epsilon \leftarrow \epsilon'$;
      $count \leftarrow count+1$ \\
      }
      $\epsilon_{T+1} \leftarrow \epsilon$\\
      \KwRet{$\epsilon_{T+1}$}
     }
    \caption{Adaptive Conservativeness Level}
    \label{alg:adaptive}
\end{algorithm}

\subsection{Convergence Analysis of Proposed Algorithm}
We first intuitively analyze the expected outcomes of Algorithm~\eqref{alg:adaptive} at each iteration $T$. Given that the partial derivative terms in \eqref{gradient} are always non-negative, the direction in which $\epsilon$ is updated, indicated by the sign of $\frac{\partial \mathcal{L}(\epsilon)}{\partial \epsilon}$, depends on  relative magnitudes of the expected versus actual cost and CVaR values. If $\text{Cost}_{\underline t}^{exp} > \text{Cost}_{\underline t}^{act}$, indicating that the DSO is overly conservative by planning for adverse scenarios that never materialize, then the value of $\epsilon$ should be decreased. Similarly, if $\text{CVaR}^{exp} > \text{CVaR}^{act}$, suggesting that the DSO has prepared excessively for constraint violations that occur less frequently in practice, then $\epsilon$ should also be reduced.

It is also possible that the two terms in \eqref{gradient} may indicate differing directions for updating $\epsilon$. This is because \eqref{mod:uncertain} considers two worst-case distributions that may not coincide: one associated with the maximum cost in \eqref{uncertain:obj} and the other corresponding to the maximum violation rate of the chance constraint in \eqref{uncertain:V_bound}. In this case, $\epsilon$ will be updated using a weighted sum of the two directions suggested by both the cost and CVaR, and the weights are the partial derivatives in \eqref{gradient_Cost} and \eqref{gradient_CVaR}. As highlighted in \cite{mieth2023data}, the values of $\lambda_{j}^{co}$ and $\lambda_{j}^{cc}$ can be interpreted as the marginal values of $\epsilon$ in the optimal solution of \eqref{mod:convex_DRO}. Therefore, the weights of the cost and the CVaR terms in \eqref{gradient} indicate their relative importance within the context of risk-averse decision-making.

Formally, we propose and prove the following proposition for the
convergence of $\epsilon_T$ based on Algorithm~\ref{alg:adaptive}:
\begin{mypro} \label{Convergence_each_iteration}
    (Convergence of $\epsilon_T$) For each iteration $T \le T^\text{out}$, Algorithm~\ref{alg:adaptive} will return a near optimal value of $\epsilon_T$, under the condition that $count^{\max}$ is adequately large, and both the learning rate $\chi$ and the convergence threshold $\Delta \epsilon^{\min}$ are set to sufficiently small values.
\end{mypro}
\begin{proof}
The essence of proving this proposition lies in showing the convexity of loss function $\mathcal{L}(\epsilon)$ with respect to $\epsilon$, in accordance with the well-established convergence theorem of the gradient descent method for convex functions \cite{boyd2004convex}. This requires an analysis of the second-order derivative:
\begin{align}
    \frac{\partial^{2} \mathcal{L}(\epsilon)}{\partial \epsilon^{2}} = 2\left(\frac{\partial \text{Cost}(\epsilon)}{\partial \epsilon}\right)^2 + 2\left(\frac{\partial \text{CVaR}(\epsilon)}{\partial \epsilon}\right)^2 \nonumber
\end{align}
which is always non-negtive. Therefore, the $\mathcal{L}(\epsilon)$ is convex in $\epsilon$, implying the convergence of the gradient-based method in  Algorithm~\ref{alg:adaptive}.
\end{proof}

The optimal value of $\epsilon$ conveys a tangible significance as it represents the weighted average of the two Wasserstein distances. The first is the distance between the empirical distribution $\hat{\mathbb{P}}_{N_s}$ and the distribution that results in $\text{Cost}{\underline t}^{act}$, while the second is the distance between $\hat{\mathbb{P}}_{N_s}$ and the distribution leading to $\text{CVaR}^{act}$. Essentially, the optimal value of $\epsilon$ reflects the distance between the empirical and true distributions of $\xi$.

\subsection{Conservativeness Change in Long Term}
Regrading the change of $\epsilon_T$ when $T$ increases, we propose and prove the following proposition:
\begin{mypro} \label{Convergence_long_term}
    (Convergence of $\epsilon_{\infty}$) Assuming the true distribution of $\xi$ is stationary (i.e., it does not change over time) and the iterations within the inner loop $t^\text{in} \to \infty$, then the optimal value of $\epsilon_{T} \to 0$ when $T \to \infty$.
\end{mypro}
\begin{proof}
 Proving this proposition entails the proofs of three key aspects: i) the optimal value of $\epsilon_T$ decreases monotonically as $T$ increases; ii) the minimal possible value for $\epsilon_T$ is 0, and iii) the optimal value of $\epsilon_T$ is stable when $T \to \infty$. First, with the assumption that $t^\text{in} \to \infty$ and the stationary nature of the true distribution, the sample distribution characterized by $\{\xi_t\}_{t =1}^{ t^\text{in}}$ effectively mirrors the true distribution of $\xi$ for each iteration $T$. As $T$ progresses, the inclusion of additional samples into the historical dataset implies that the Wasserstein distance between the empirical distribution $\hat{\mathbb{P}}_{N_s}$ and the true distribution should decrease monotonically.
 Second, $\epsilon = 0$ implies a perfect match between the empirical and true distributions of $\xi$, which is achievable through the learning algorithm outlined in Section~\ref{sec:learning} given sufficient learning iterations $T$. In such scenarios, the expected and actual costs, as well as the expected and actual CVaR, converge and thus  yield a zero gradient in \eqref{gradient}. This outcome indicates the stability of $\epsilon_{\infty}$, thereby proving the third point.
\end{proof}

However, the two assumptions made in Proposition~\ref{Convergence_long_term} might not always apply in real-world scenarios. First, the inner loop is limited to a finite number of iterations, which means the sample distribution might not accurately represent the true distribution of $\xi$. Consequently, gradient estimates could be imprecise or noisy due to the reliance on a limited number of samples, causing fluctuations in the convergence trajectory. This scenario bears resemblance to the behavior of stochastic gradient descent (SGD) as opposed to gradient descent (GD), where the former is known for its variability in the path towards convergence due to its sample-based gradient updates.

On the other hand, the true distribution of $\xi$ may not be stationary due to factors like the installation of new DERs, which may change nodal flexibility regions, or fluctuations in residential electricity prices which may affect DER participation costs. In this case, the empirical distribution learned by the DSO may not perfectly match the true distribution, indicating that the optimal value of $\epsilon$ will not converge to 0. Regarding this, we have the following remark:
\begin{myremark} \label{stationary}
   If changes in the true distribution of $\xi$ are slow and smooth, the optimal value of $\epsilon_T$ is expected to stabilize at a small positive value when $T \to \infty$. This reflects the rate of change in the true distribution.
\end{myremark}

Proposition~\ref{Convergence_long_term} also informs what steps should decision-makers take once $\epsilon_T$ reaches 0. With the optimal level of conservativeness being 0, a transition from a risk-averse to a risk-neutral stance is recommended.
This shift involves minimizing the expected cost rather than the worst-case costs in the objective and considering the expected rate of constraint violations instead of the worst-case violation rate.

If $T$ keeps increasing, more samples will be added into the empirical distribution without providing new information, it may be appropriate to stop the learning process. Furthermore, given that the computational complexity of a sample-based decision-making model depends on the number of samples, reducing the sample count through sample selection techniques 
or diminishing the size of the uncertainty set could be considered to improve the efficiency.

\section{Numerical Experiments}
\label{sec:Case_study}
This section considers interactions between one DSO and one DERA to show the effectiveness of the proposed method. We set the initial value of $\epsilon =0.01$ and the learning rate $\chi = 0.001$. The iteration number in the inner loop is $t^\text{in}=5$, and the true distribution of $\xi$ is assumed to follow a uniform distribution between -1 and 2. 
Furthermore, We assume that initially the DSO has access to 10 samples of $\xi$. This initial count, set lower than what might be found in real-world situations, is chosen to more effectively highlight the convergence trend of $\epsilon$ as the number of available samples increases.
All the numerical experiments are conducted using the Gurobi solver \cite{gurobi} in Python and are performed on a PC with an Intel Core i9, 2.20 GHz with 16 GB RAM. 

We first examine the convergence of the optimal conservativeness level $\epsilon$ when $\epsilon_T$ is obtained using Algorithm~\ref{alg:adaptive} at each iteration $T$. Fig.~\ref{fig:convergence} (a) shows the value of $\epsilon_T$ as $T$ increases, where the black line is a representative convergence trajectory, and the shaded area represents the confidence interval derived from conducting the experiment for 50 times. It is clear that $\epsilon_T$ converges to 0, with some fluctuations along the convergence trajectory. These fluctuations are attributed to using only five samples of $\xi$ in the inner loop to approximate its true distribution.

\begin{figure}[htbp]
    \centering
    \includegraphics[width=1\linewidth]{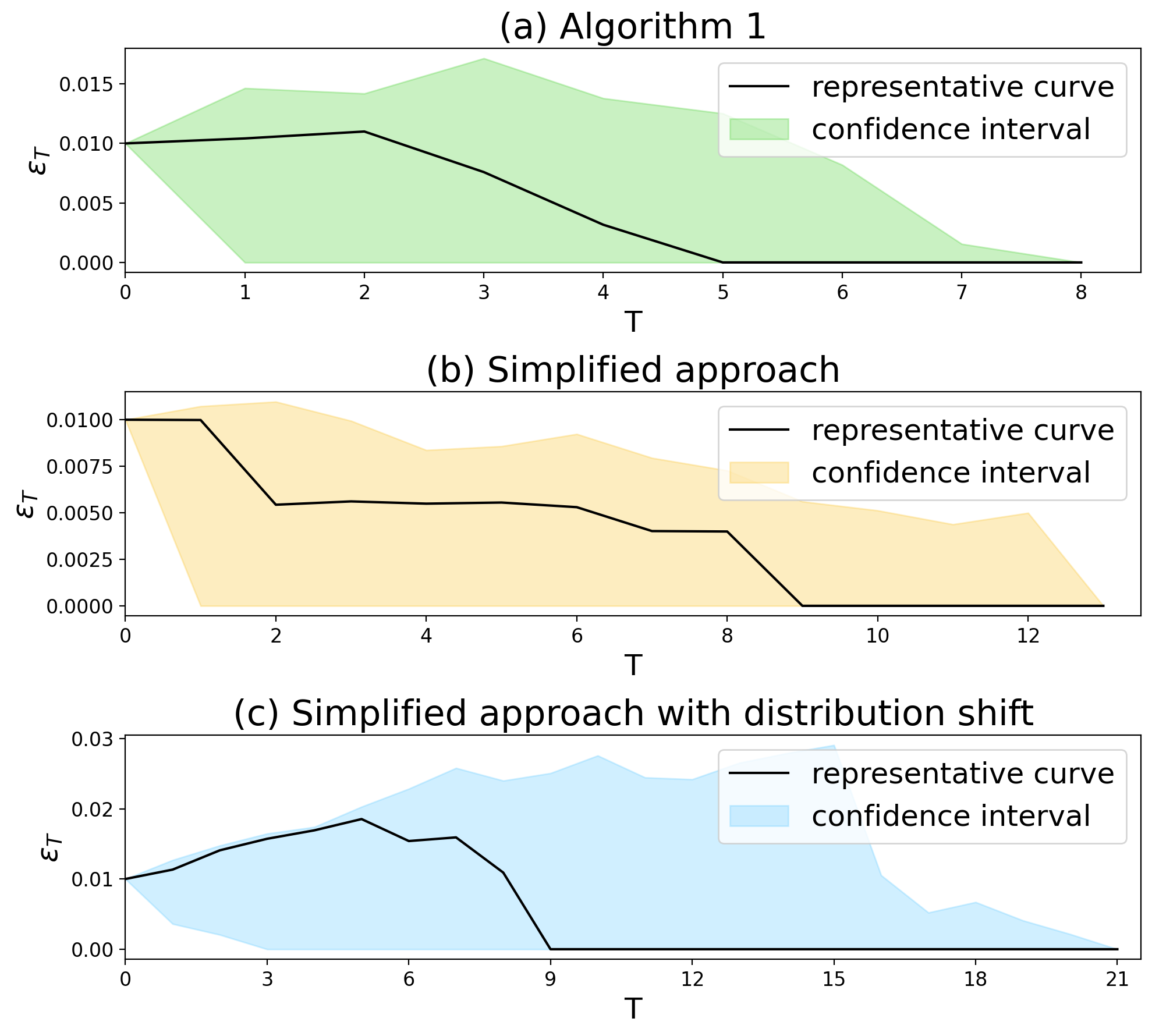}
    \caption{Convergence trajectory of optimal conservativeness level.}
    \label{fig:convergence}
\end{figure}

The computational challenge presented by Algorithm~\ref{alg:adaptive} stems from iteratively solving the distributionally robust incentive design problem in \eqref{mod:convex_DRO}. A practical simplification approach involves updating $\epsilon$ just once per iteration $T$ using \eqref{updating_rule}, rather than continuously updating until the optimal value of $\epsilon_T$ is reached. Fig.\ref{fig:convergence} (b) illustrates the convergence trajectory of $\epsilon_T$ using this simplified approach. While the convergence is slower compared to the trajectory shown in Fig.\ref{fig:convergence} (a), $\epsilon_T$ still approaches 0 after several iterations. This highlights a viable pathway for enhancing the computational efficiency of the proposed algorithm.

The robustness of this simplified method is further evaluated in a more realistic scenario with shifts in the true distribution of $\xi$. At each iteration $T$, the DERA responses in the inner loop are generated based on a varying distribution of $\xi$. More specifically, the bounds of the uniform distribution of $\xi$, initially set at -1 and 2, are expanded at a constant rate of 0.1 in every iteration. Fig.~\ref{fig:convergence} (c) shows that $\epsilon$ converges to 0 despite the distribution shift. 

\section{Conclusion and Future Work}
\label{sec:conclusion}
This paper proposes to learn parameters of the distributionally robust optimization applied to incentive design for voltage regulation 
 with adjustable conservativeness.  The incentive design problem is reformulated as a convex model, facilitating efficient resolution with standard solvers. We also prove the convergence of the gradient-based method for updating the conservativeness level. Additionally, we analyze the optimal conservativeness level, establishing a connection between risk-averse decision-making and the learning process. Numerical experiments on a single DSO and a single DERA scenario shows the effectiveness of our proposed approach. 

One limitation of the proposed method is its assumption of a fully rational DERA, meaning it presumes the DER responses are dictated by an optimization model. In future research, we plan to employ a more realistic decision-making model for DERAs, for instance, by incorporating the concept of bounded rationality \cite{marin2021electric}.


\bibliographystyle{IEEEtran}
\bibliography{literature}

\end{document}